\documentclass[letterpaper,10pt,conference]{ieeeconf}

\pdfoutput=1

\IEEEoverridecommandlockouts
\usepackage[pdftex]{graphicx}
\DeclareGraphicsExtensions{.pdf,.jpg,.png,.mps}

\usepackage{amsmath,amssymb}   
\usepackage[utf8]{inputenc} 
\usepackage[english]{babel}

\newcommand{\set}[1]{\mathcal{#1}}
\newcommand{\op}[1]{\mathrm{#1}}

\usepackage{amsthm}
\newtheorem{lemma}{Lemma}
\newtheorem{theorem}{Theorem}
\newtheorem{problem}{Problem}

\theoremstyle{definition}
\newtheorem{example}{Example}

\usepackage{cleveref}
\usepackage{color}

\begin{document}

\title{Fuel-Optimal Centralized Coordination of Truck Platooning\\Based on Shortest Paths}
\author{Sebastian van de Hoef, Karl H. Johansson and Dimos V. Dimarogonas
\thanks{The authors are with the ACCESS Linnaeus Center and the School of Electrical Engineering, KTH Royal Institute of Technology, SE-100 44, Stockholm, Sweden. {\tt \{shvdh, kallej, dimos\}@kth.se}. This work was supported by the COMPANION EU project.
}
}

\maketitle

\begin{abstract}
 Platooning is a way to significantly reduce fuel consumption of trucks. Vehicles that drive at close inter-vehicle distance assisted by automatic controllers experience substantially lower air-drag. In this paper, we deal with the problem of coordinating the formation and the breakup of platoons in a fuel-optimal way. We formulate an optimization problem which accounts for routing, speed-dependent fuel consumption, and platooning decisions. An algorithm to obtain an approximate solution to the problem is presented. It first determines the shortest path for each truck. Then, possible platoon configurations are identified. For a certain platoon configuration the optimal speed profile is the solution to a convex program. The algorithm is illustrated by a realistic example. 
\end{abstract}

\section{Introduction}

Truck-platooning means that several trucks drive with small inter-vehicle distance. The distance is maintained by appropriate, automatic control of the vehicle speeds.
Earlier research on truck-platooning has mostly been motivated by the possible increase in traffic throughput \cite{frameworks_overview_article,path_overview_conference,dolphin_framework_conference,Halle_car_platoons_multiagent_auto21,Hall_platoon_sorting}. The potential for fuel saving has renewed the interest in the topic\cite{assad_thesis}. When the spacing between the vehicles is sufficiently small, the follower vehicles experience reduced aerodynamic drag \cite{air_drag_reduction_platooning, Bonnet2000}.
This, in turn, leads to reduced fuel consumption. Advances in related technologies have made the commercial availability of trucks that can platoon likely in near future. 

Much of the research on truck-platooning focuses on the control of the inter-vehicle distances 
\cite{string_stability,barooah_platooning}.
Unless trucks have the same route and start at the same time, platoons have to be created during their journeys. If we use platooning to reduce fuel consumption and only have a small part of vehicles equipped with a platooning system, the formation of the platoons needs to be coordinated over large regions, which is the subject of this work.
Related problems arise for example in air traffic management \cite{air_traffic_management}, convoy movement \cite{convoy_movement_opt}, and other areas  \cite{pinedo2012scheduling,uav_beard,Bullofrazzoli_vehicle_routing}.
Related research on coordination of truck platooning is considered in \cite{jeff_distributed_coordination} where a heuristic scheme based on distributed controllers at road intersections is proposed. The complexity of an off-line optimization algorithm is discussed in \cite{jeff_complexity}. Schemes where platoons are formed on the on-ramps as in \cite{Hall_platoon_sorting} require huge investments in infrastructure. The authors of \cite{datamining_platooning} consider that trucks wait to form platoons and use data-mining techniques. \cite{its_energy_project} computes the global fuel saving opportunities based on a simulation study but does not explicitly detail the means of coordination.

The goal of this paper is to introduce a framework for a fuel-optimal coordination where the trucks adapt their speeds in order to form platoons during their journeys while driving. The influence of speed on the fuel consumption is explicitly considered and trucks are guaranteed to arrive at their destinations by their arrival deadlines.

We consider a number of transport assignments, each with a start position and time and an arrival position and deadline. In Section~\ref{sec:problem_statement} we introduce notation and formulate an optimization problem. The route of a truck is modeled as a path in a graph. The average speed of the trucks is constrained to take into account legal speed limits, performance of the truck, and traffic. Decision variables are the paths taken from start to destination and the speed on each edge. 
We optimize the total fuel consumption for all transport assignments while taking into account both reduced fuel consumption due to platooning as well as speed-dependent fuel consumption. Platooning is considered to take place when two or more trucks traverse an edge at the same time with the same speed. We illustrate why this optimization problem is difficult to solve.
In Section~\ref{sec:approx_sol}, we derive an approximate solution to the problem. The approximate solution is based on the shortest path for each transport assignment, which can be efficiently computed with available algorithms. The remaining selection of edge traversal speeds is solved by identifying platoon configurations taking into account the route and the timing constraints. For a certain platoon configuration, the optimal speeds are computed as the solution to a convex, continuous optimization problem, which can be efficiently solved. We describe how the existence of a solution to the convex problem can be checked for a platoon configuration. By searching through the possible platoon configurations, we obtain optimal speed profiles for the trucks. We show that each two trucks form a platoon at most once during a journey, which can significantly reduce the discrete search space. We demonstrate the algorithm on a realistic scenario.

\section{Problem Statement}
\label{sec:problem_statement}

We consider a road network modeled as a weighted, directed graph $\set{G} = (\set{N},\set{E},W)$, where $\set{N}$ is a set of nodes, $\set{E} \in \set{N} \times \set{N}$ is a set of edges, and $W: \set{E} \rightarrow \mathbb{R}^+$ are positive edge weights. Each edge corresponds to a road segment and the nodes correspond to the intersections. The edge-weights model the length of the road segments. 

We have $K$ transport assignments, each characterized by a start node $n_k^S \in \set{N}$ and a destination node $n_k^D \in \set{N}$. There is a start-time $t_k^S$ at which the transport starts from $n_k^S$ and a deadline $t_k^D$ by which the transport needs to reach $n_k^D$. We consider that each transport assignment is assigned to one truck. The truck carrying out transport assignment $k$ has index $k$, i.e., the same index as the transport assignment.

The speed of a truck is positive and smaller than an allowed maximum speed $v_{\max}$. We assume that the speed over each edge is constant. So, if the speed of a truck over edge $(n_1,n_2)$ is $v$ and the truck starts traversing the edge at time $t$, it arrives at $t + W((n_1,n_2))/v$ at node $n_2$. 

A path of $\set{G}$ is a sequence of nodes $(n_k[1],n_k[2],\dots)$, so that $(n_k[i],n_k[i+1]) \in \set{E}$ for $i = 1,\dots,(|n_k|-1)$, where $|n_k|$ denotes the length of the sequence, i.e., the number of nodes in the path. A path connecting $n_k^S$ and $n_k^D$ has $n_k[1] = n_k^S$ and $n_k[|n_k|] = n_k^D$.

We say that a set of edges $\tilde{\set{E}}$ forms a path if and only if there exists a path $\tilde{n}$ so that $\tilde{\set{E}} = \{(\tilde{n}[i],\tilde{n}[i+1]): i = 1,\dots,(|\tilde{n}|-1)\}$. 
Let $\tilde{n}$ be a path. We say that $\tilde{n}_s$ is a sub-path of $\tilde{n}$ if and only if $\tilde{n}_s$ is a path and there exists $j$ such that for every $i = 1,\dots,(|\tilde{n}_s|-1)$ it holds that $(\tilde{n}_s[i],\tilde{n}_s[i+1]) = (\tilde{n}[j + i],\tilde{n}[j+i+1])$.
We denote the sequence of edges on the path of transport assignment $k$ as $e_k[i] = (n_k[i],n_k[i+1])$ for  $i = 1,\dots,(|n_k|-1)$.

For a path $n_k$, we can collect the edge traversal speeds in a sequence $v_k$ of length $|e_k|$. The $i$th element $v_k[i]$ is the edge traversal speed on edge $e_k[i]$.

The arrival time $t_k[i]$ of a truck $k$ at the $i$th node on its path is the sum of the start time and the edge traversal times up to this point. Given the sequences $n_k$ and $v_k$, we calculate $t_k$ recursively as
\begin{equation}
\begin{split}
 t_k[1] &= t_k^S\\
 t_k[i] &= t_k[i-1] + \frac{W(e_k[i-1])}{v_k[i-1]},
 \; i \in \{2,\dots,|n_k|\}. \label{eq:t_k_def}
\end{split}
\end{equation}
We denote the set of paths as $\set{R} = \{n_1,\dots,n_K\}$ and the set of speed sequences as $\set{V} = \{v_1,\dots,v_K\}$.

\subsection{Fuel Consumption Model}
\label{sec:fuel_model}

We model the fuel consumption per distance traveled as a function of speed $f(v)$. A simple model of fuel consumption per distance is a second order polynomial in the speed:
$
 f(v) = F_r + F_a v^2,
$
where $F_r, F_a$ are positive constants \cite{assad_thesis}. $F_r$ accounts for forces mainly caused by rolling resistance and $F_a v^2$ accounts for aerodynamic forces. For the sake of simplicity, we neglect other factors influencing the fuel consumption such as selected gear, road grade, wind, etc. Many of these influences could be added without fundamentally changing the problem formulation, but they would complicate notation.

When platooning as a follower, i.e., not the first vehicle in the platoon, we assume that the aerodynamic coefficient is reduced to $\eta F_a$ with $0 < \eta \leq 1$, i.e.,
$
 f(v,\eta) = F_r + \eta F_a v^2.
$
We neglect the small fuel savings of the platoon leader, i.e., the first vehicle in the platoon. Note that $\eta = 1$ corresponds to the fuel consumption without platooning. Recent experiments have shown that, typically, $\eta \approx 0.6$ \cite{Bonnet2000}. We say that a group of trucks platoon over edge $(n_1,n_2)$, if the start times of the edge traversal and the speeds on the edge are the same for all trucks in the group. One of the trucks is assigned to be the platoon leader. The coefficients $F_r$ and $F_a$ are assumed to be the same for all trucks. This implies that $\eta$ for the followers does not depend on the platoon leader or on how the trucks are ordered in the platoon. We appoint the truck in the platoon with the highest index the role of the platoon leader.

The complete fuel expense is the sum of the fuel expenses of each truck on each edge on its path:
$
 f_c(\set{R},\set{V}) = \sum\limits_{k = 1}^{K} \sum\limits_{i = 1}^{|e_k|} W\big(e_k[i]\big) f\big(v_k[i],\eta_l(k,i)\big),
$
with
$\eta_l(k,i) = \eta$ if $\exists\, k_l,i_l: k_l > k, e_{k_l}[i_l] = e_k[i], t_{k_l}[i_l] = t_k[i], v_{k_l}[i_l] = v_k[i]$ and $\eta_l(k,i) = 1$ otherwise.
$\set{R}$, $\set{V}$ is the set of paths and the set of speed sequences respectively.
The function $\eta_l$ is equal to $\eta$, if a truck with higher index $k_l$ traverses the same edge $e_{k_l}[i_l] = e_k[i]$ at the same time $t_{k_l}[i_l] = t_k[i]$ with the same speed $v_{k_l}[i_l] = v_k[i]$, and 1 otherwise. Hence, $\eta_l = \eta$ if the current truck gets the role of a platoon follower and then experiences reduced air-drag on this edge, and $\eta =1$ otherwise.

\subsection{Problem Statement}

We are ready to state the optimization problem. For each transport assignment $k$ we want to find a path in $\set{G}$ from $n_k^S$ to $n_k^D$ denoted by $n_k$ and a speed $v_k$ for each edge on the path so that the truck arrives in time, does not violate the speed constraints, and that the total fuel expense $f_c$ is minimized.
\begin{problem}\label{prob:basic_problem}
\begin{subequations}
\begin{align}
 &\min\limits_{ \set{R},\set{V}} f_c(\set{R},\set{V}) \label{eq:opt_prob_obj}\\
 &\;\;\;\op{ s.t. }\nonumber\\
 &\op{for} \; k \in \{1,\dots,K\}\nonumber\\
 &\;\; 0 \leq v \leq v_{\max},\; v \in v_k \label{eq:opt_prob_v_con}\\
 &\;\; n_k[1] = n_k^S,\; n_k[|n_k|] = n_k^D && \label{eq:opt_prob_n_1_con}\\
 &\;\; e_k[j] \in \set{E},\; j \in \{1,\dots,|e_k|\} \label{eq:opt_prob_n_path_con}\\
 &\;\; t_k[1] = t_k^S && \label{eq:opt_prob_t_s}\\  
 &\;\; t_k[|n_k|] \leq t_k^D && \label{eq:opt_prob_t_d_con}
\end{align}
\end{subequations}
\end{problem}
The constraints \eqref{eq:opt_prob_v_con} ensure that the speed $v$ on each edge is positive and below the maximum speed $v_{\max}$. The constraints \eqref{eq:opt_prob_n_1_con} ensure that the paths connect the start and the destination of each transport assignment. The constraints \eqref{eq:opt_prob_n_path_con} ensure that each sequence of node $n_1,\dots,n_k$ describes a path in $\set{G}$, i.e., that between any two successive nodes in $n_k$ there is an edge in $\set{E}$ connecting these nodes. The constraints \eqref{eq:opt_prob_t_s} ensure that each truck starts at $t_k^S$. The constraints \eqref{eq:opt_prob_t_d_con} ensure that each truck arrives in time with $t_k$ given by \eqref{eq:t_k_def}.
It would not be difficult to make $v_{\max}$, $F_r$, and $F_a$ dependent on the edge. We omitted this, however, in order to keep the notation simple. 
In general, Problem~\ref{prob:basic_problem} is not straightforward to solve, since small changes in the constraints can entirely change the routes of the trucks when there are several possible routes of similar length.

\section{Approximate Solution Based on Shortest Paths}
\label{sec:approx_sol}

The problem stated in Section~\ref{sec:problem_statement} is hard to solve exactly \cite{jeff_distributed_coordination,jeff_complexity}. Therefore, we propose here an approximate solution which addresses routing and speed selection in two successive steps. 

In the first step, we determine the paths $\set{R}$ \cite{shortest_path_overview}. We choose $n_k$ as the shortest paths with respect to the weights $W$ between $n_k^S$ and $n_k^D$. We assume that the shortest path for any pair of nodes in $\set{N}$ is unique.

Having computed the path for each truck, the problem remains to select the speeds on each edge of the paths such that $f_c$ is minimized. Without the fuel saving effect of platooning, this would just be the slowest speed for each truck with which it arrives before the deadline $t_k^D$. Taking platooning into account, it might, however, be beneficial to choose higher speeds on some edges, in order to form platoons. There is a non-trivial trade-off between higher fuel consumption due to increased speed and reduced fuel consumption due to platooning. Notice that the speed-independent part of the fuel consumption
$
 \sum\limits_{k = 1}^{K} \sum\limits_{i = 1}^{|e_k|} W(e_k[i]) F_r
$
is fixed, if the paths are fixed, and hence it does not have to be considered when selecting the speeds. 

We formulate the approximate problem with $\set{R}$ ($e_k$) as the shortest routes:
\begin{problem}\label{prob:shortest_path_problem}
\begin{subequations}
\begin{align}
 &\min\limits_{\set{V}} \sum\limits_{k = 1}^{K} \sum\limits_{i = 1}^{|e_k|} W\big(e_k[i]\big) F_a v_k[i]^2\eta_l(k,i) \label{eq:opt_prob_obj_path}\\
 &\;\;\;\op{s.t.}\nonumber\\
 &\op{for}\; k \in \{1,\dots,K\}\nonumber\\
 &\;\; 0 \leq v \leq v_{\max},\; v \in v_k \label{eq:opt_prob_v_con_path}\\
 &\;\; t_k[1] = t_k^S \label{eq:opt_prob_t_s_path}\\  
 &\;\; t_k[|t_k|] \leq t_k^D && \label{eq:opt_prob_t_d_con_path}
\end{align}
\end{subequations}
\end{problem}

We cannot directly apply standard continuous optimization techniques to Problem~\ref{prob:shortest_path_problem} since $f_c(\set{R},\set{V})$ is not continuous in $\set{V}$. Instead, we propose explicitly considering which trucks platoon together over which edges as a discrete decision, which we call \textit{platoon configuration}. Based on this decision, we can formulate a continuous, convex optimization problem with linear constraints, whose solution gives $\set{V}$ in a way that $f_c(\set{R},\set{V})$ is minimal for this particular platoon configuration. In order to find the optimal platoon configuration, we solve a convex program for each feasible platoon configuration and select the best one. 

\subsection{Optimal Speed for a Given Platoon Configuration}

In this section, we will describe how to model the platoon configuration and how to formulate the convex program whose optimal solution are the optimal velocities $\set{V}$ for a given platoon configuration. A program like this can be efficiently solved. In Section~\ref{sec:platooning_posibilities}, we will describe how to search through the possible platoon configurations in order to determine the optimal one.

Firstly, we specify in mathematical terms what we mean by \textit{platoon configuration}. We assign each transport assignment $k$ on each edge of its path $n_k$ a \textit{predecessor}. If the truck is not part of a platoon or if it is a platoon leader on this particular edge, the predecessor is the truck itself. If it is a follower in a platoon, the predecessor is the truck in the platoon with the smallest index higher than $k$. We collect the predecessors for truck $k$ in a sequence $l_k$. 

We reformulate the problem in terms of the traversal times $T_k[i] = W(e_k[i])/v_k[i]$ as opposed to the traversal speeds $v_k$ since this allows us to express all constraints as linear constraints. The choice of $l_k$ determines $\eta_l$ to $\eta_l(k,i) = \eta$ if $l_k[i] \neq k$. We will assume that $l_k$ is chosen so that $\eta_l(k,i) = 1$ if $l_k[i] = k$.
We can write the objective function as
\begin{equation}
 F_a \sum\limits_{k=1}^{K} \sum\limits_{i=1}^{|e_k|} \eta(k,i) \frac{W(e_k[i])^3}{T_k[i]^2}, \label{eq:red_objective_f}
\end{equation}
where we have only kept the speed dependent part of $f_c$ and substituted $v_k[i]^2 = W(e_k[i])^2/T_k[i]^2$. We can also remove $F_a$ from the objective function since it is a constant positive scaling factor in \eqref{eq:red_objective_f}.

Constraint \eqref{eq:opt_prob_v_con_path} is reformulated in terms of $T_k$ as
$
 0 \leq W(e_k[i])/v_{\max} \leq T_k[i].
$
Constraints \eqref{eq:opt_prob_t_s_path} and \eqref{eq:opt_prob_t_d_con_path} turn into
$
 t_k^S + \sum\limits_{i=1}^{|e_k|} T_k[i] \leq t_k^D.
$
Finally, we have to introduce a number of constraints that ensure that the trucks do platoon as specified by $l_k$. Therefore, we have to add a constraint for each truck $k \in \{1,\dots,K\}$ on each link $j \in \{1,\dots,|e_k|\}$ on its path, ensuring that it arrives at the same time at the beginning of the edge as its predecessor: $t_k^S + \sum\limits_{i = 1}^{j-1} T_k[i] = t_{l_k[j]}^S + \sum\limits_{i = 1}^{j_l-1} T_{l_k[j]}[i]$, and that it traverses the edge at the same speed: $T_k[j] = T_{l_k[j]}[j_l]$, where we define $j_l: n_k[j] = n_{l_k[j]}[j_l]$ as the index of node $n_k[j]$ in the path of the predecessor of truck $k$ on edge $e_k[j]$.
We get the following optimization problem for a given platoon configuration $l_1,\dots,l_K$.
\begin{problem}\label{prop:convex_cont_problem}
\begin{subequations}
 \begin{align}
 &\min\limits_{ \{T_1,\dots,T_K\}} \sum\limits_{k=1}^{K} \sum\limits_{i=1}^{|e_k|} \eta(k,i) \frac{W(e_k[i])^3}{T_k[i]^2} \\
 &\;\;\;\mathrm{s.t.}\nonumber\\
 &\mathrm{for} \; k \in \{1,\dots,K\}\nonumber\\
 &\;\; \frac{W(e_k[i])}{v_{\max}} \leq T_k[i], \; i \in \{1,\dots,|e_k|\} \label{eq:cont_opt_v_max_con}\\
 &\;\; t_k^S + \sum\limits_{i=1}^{|e_k|} T_k[i] \leq t_k^D \label{eq:cont_opt_t_d_con}\\
 &\;\; t_k^S + \sum\limits_{i = 1}^{j-1} T_k[i] = t_{l_k[j]}^S + \sum\limits_{i = 1}^{j_l-1} T_{l_k[j]}[i], 
 \; j \in \{1,\dots,|e_k|\} \label{eq:cont_opt_plat_arr_con}\\
 &\;\; T_k[j] = T_{l_k[j]}[j_l],
 \; j \in \{1,\dots,|e_k|\}, \label{eq:cont_opt_plat_trav_con}
\end{align}
\end{subequations}
with $j_l: n_k[j] = n_{l_k[j]}[j_l]$
\end{problem}
We can verify that we have a convex, twice-differentiable objective function on the domain of positive traversal times $T_k[i] > 0: k \in \{1,\dots,K\},\; i \in \{1,\dots,|T_k|\}$ and constraints that are linear in the decision variables $\{T_1,\dots,T_K\}$. 
We provide an example to illustrate the formulation and solution of Problem~\ref{prop:convex_cont_problem}.
\begin{example}
Figure~\ref{fig:example_graph_cont_problem} depicts a small road network and the paths of two trucks. We choose for the platoon configuration that the two trucks platoon from node 3 to 5, which means $l_1=(1,2,2,1,1)$ and  $l_2=(2,2,2,2,2)$.

Figure~\ref{fig:two_truck_cont_opt} shows the solution to this problem. We can see that the two trucks do meet at node 3 and platoon until node 5. Then they split up and have different velocities due to their different deadlines. An interesting observation is that both trucks attain their maximum speed, corresponding to minimum slope in Figure~\ref{fig:two_truck_cont_opt}, while platooning whereas lower speeds are optimal when they drive on their own. This is due to the reduced air drag that can render beneficial to choose high speed while driving in a platoon and low speed otherwise. 
\end{example}

\begin{figure}
  \begin{center}
 \includegraphics[width = 0.8\columnwidth]{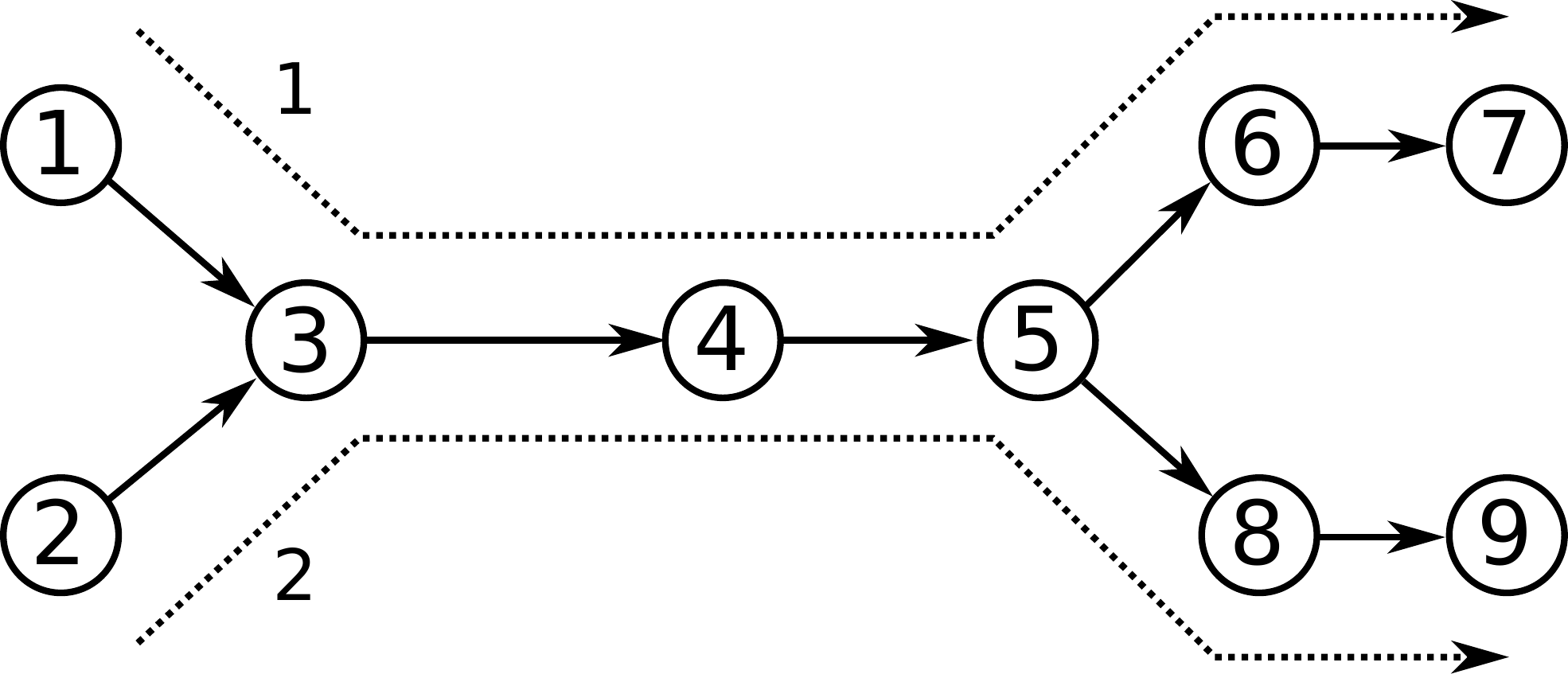}
 \caption{
 Example road network to illustrate Problem~\ref{prop:convex_cont_problem}. The nodes are indicated by circles and directed edges by arrows. The edge between node ``3'' and ``4'' has length 2, all other edges length 1. The dashed arrows indicate the paths of the trucks.
}
\label{fig:example_graph_cont_problem}
\end{center}
\end{figure}

\begin{figure}
  \begin{center}
 \includegraphics[width = 1.\columnwidth]{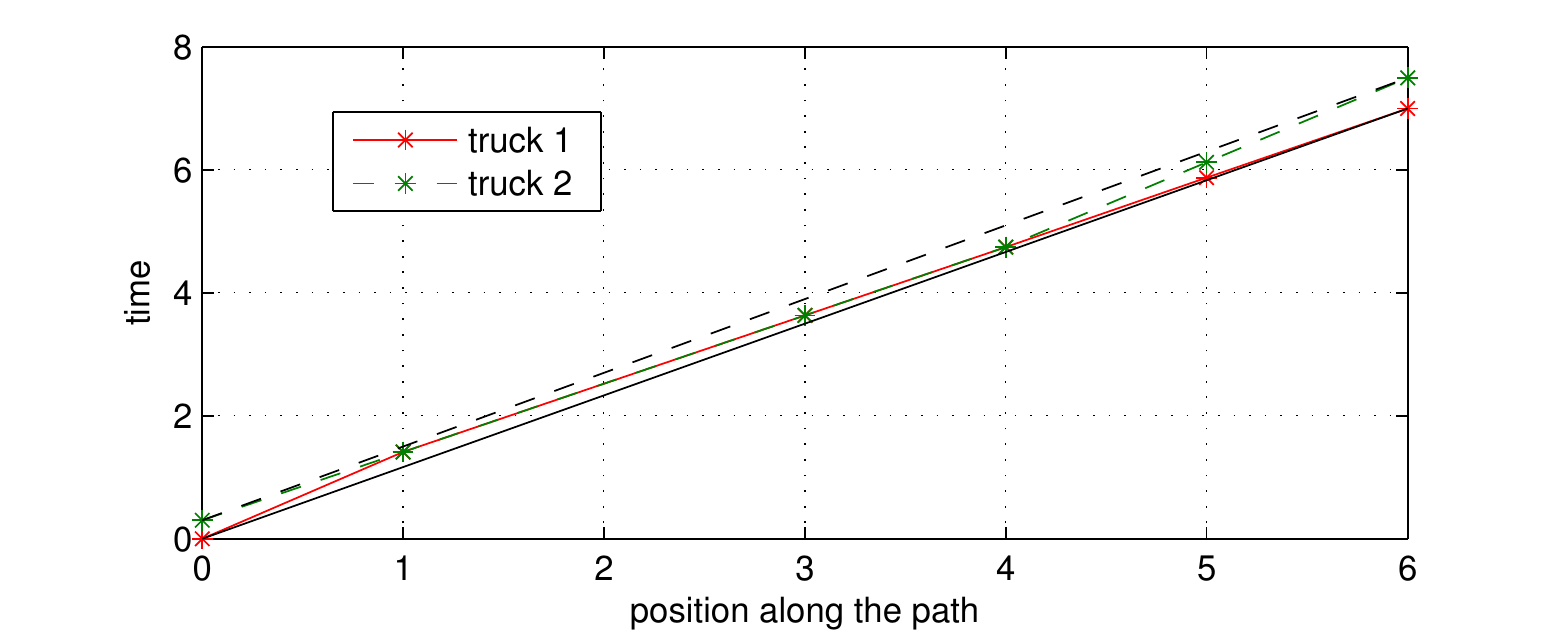}
 \caption{Optimal solution to Problem~\ref{prop:convex_cont_problem} for the road network and the paths shown in Figure~\ref{fig:example_graph_cont_problem}. The time trajectory along the route is shown. Smaller slope corresponds to higher speed. The crosses indicate the nodes' positions. The optimal trajectories without the platooning effect are drawn in black. }
\label{fig:two_truck_cont_opt}
\end{center}
\end{figure}

\subsection{Evaluation of Platooning Possibilities}
\label{sec:platooning_posibilities}

In order to find the optimal platoon configuration, we propose solving Problem~\ref{prop:convex_cont_problem} for different platoon configurations. We first describe how to enumerate all platoon configurations so that Problem~\ref{prop:convex_cont_problem} has a feasible solution. We then establish a property of optimal platoon configurations which further limits the discrete search space. Finally, we provide examples to illustrate the procedure. 

In the first step, we determine which paths share edges. Two trucks can only platoon on common edges on their paths. We define the set of shared edges between two paths $n_{k_1}$, $n_{k_2}$ as
$
n_{k_1} \cap_\set{E} n_{k_2} = \{(n_{k_1}[i],n_{k_1}[i+1]): \exists j: (n_{k_1}[i],n_{k_1}[i+1]) = (n_{k_2}[j],n_{k_2}[j+1]) \}.
$
The following lemma is useful when determining the possible platoon configurations.
\begin{lemma}\label{prop:path_intesection}
If the paths $n_{k_1}$ and $n_{k_2}$ are unique shortest paths in $\set{G}$, then either $n_{k_1} \cap_\set{E} n_{k_2}$ is empty or the path formed by $n_{k_1} \cap_\set{E} n_{k_2}$ also forms a shortest path between two nodes of $\set{G}$.
\end{lemma}
\begin{proof}
It is well known that every sub-path of a shortest path is a shortest path \cite{introduction_to_algorithms}. But then the path between any two nodes of $n_{k_1}$ as well as between any two nodes of $n_{k_2}$ is unique. If $n_{k_1} \cap_\set{E} n_{k_2}$ does not form a shortest path there exists a sub-path in $n_{k_1}$ as well as in $n_{k_2}$ connecting the same nodes which is not a sub-path of the path formed by $n_{k_1} \cap_\set{E} n_{k_2}$. Then there are different sub-paths in $n_{k_1}$ and $n_{k_2}$ connecting the same nodes. This, however, contradicts the assumption that the shortest path between two nodes is unique.
\end{proof}
If the paths of two trucks share common edges, we can check if it is possible that the two trucks form a platoon according to the constraints on the maximum speed \eqref{eq:cont_opt_v_max_con} and the before-deadline arrival \eqref{eq:cont_opt_t_d_con}. 

Therefore, we calculate the earliest and latest arrival time at each node of the paths $\set{R}$. The earliest arrival times for the path of truck $k$, denoted by $\underline{t}_k$, can be calculated forward the same way $t_k$ is calculated in \eqref{eq:t_k_def} with $v_k[1] = v_k[2] = \dots = v_k[|v_k|] = v_{\max}$: $\underline{t}_k[1] = t_k^S$ and $\underline{t}_k[i] = t_k[i-1] + W(e_k[i-1])/v_{\max}$ for $i \in \{2,\dots,|n_k|\}$.

In a similar way, we can calculate the latest arrival time $\bar{t}_k$ at each node backward. We have $t_k[1] = t_k^S$. Since there is no immediate constraint on how large $T_k[1]$ can be, the latest arrival time at nodes $n_k[2], n_k[3], \dots$ is only constrained by \eqref{eq:cont_opt_v_max_con} and \eqref{eq:cont_opt_t_d_con}. The latest arrival time $n_k[|n_k|]$ is $\bar{t}_k[|n_k|] = t_k^D$. Then, the latest time the truck can leave at $n_k[|n_k|-1]$ in order to arrive at $n_k[|n_k|]$ before $t_k^D$ is given by $\bar{t}_k[|n_k|] - W(e_k[|n_k|-1])/v_{\max}$. In this way we can work backwards until $n_k[2]$. So, we define $\bar{t}_k$ recursively as $\bar{t}_k[|n_k|] = t_k^D$, $\bar{t}_k[1] = t_k^S$ and $\bar{t}_k[i-1] = \bar{t}_k[i] - W(e_k[i-1])/v_{\max}$ for $i \in \{3,\dots,|n_k|\}$.

If there is a solution to Problem \ref{prob:shortest_path_problem}, we have $\bar{t}_k[i] \geq \underline{t}_k[i]$, and since all elements of $v_k$ can take any value in the interval $[0,v_{\max}]$, the truck can reach any $t_k[i] \in \big[\underline{t}_k[i], \bar{t}_k[i]\big]$ for a particular node $n_k[i]$. Due to \eqref{eq:cont_opt_v_max_con} not every sequence $t_k[i]$ which fulfills $t_k[i] \in \big[\underline{t}_k[i], \bar{t}_k[i]\big]$ has to be feasible. 
Figure~\ref{fig:graph_timining_constraints} illustrates the calculation of the sequences $\underline{t}_k$ and $\bar{t}_k$. The following result can now be stated.

\begin{figure}
\def\svgwidth{.8\columnwidth}
    \begin{center}
    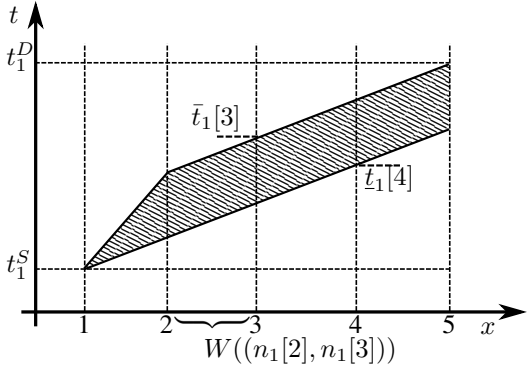
    \end{center}
\caption{Illustration of the calculation of $\bar{t}_1$ and $\underline{t}_1$ for a route with index $k=1$. On the vertical axis, time is plotted over the position along the route. The upper solid line indicates the latest and the lower solid line the earliest feasible arrival time. The position of the nodes is indicated by vertical dashed lines. The distance between two nodes is equal to the weight of the edge connecting these two nodes. As an example, $\bar{t}_1[3]$ and $\underline{t}_1[4]$ are indicated.}
\label{fig:graph_timining_constraints}
\end{figure}

\begin{theorem}\label{prop:possible_l_per_link}
Assume that $k_2 > k_1$ and that there exists a solution to Problem~\ref{prop:convex_cont_problem} without platoon followers, i.e., $l_k[\tilde{j}] = k$ for all $\tilde{j} \in \{1,\dots,|l_k|\}$ and all $k \in \{1,\dots,K\}$. Then, there exists a solution to Problem~\ref{prop:convex_cont_problem} with $k_2 = l_{k_1}[j]$ with $j \in \{1,\dots,|l_{k_1}|\}$ if and only if $\big[\underline{t}_{k_1}[j], \bar{t}_{k_1}[j]\big] \cap \big[\underline{t}_{k_2}[j_l], \bar{t}_{k_2}[j_l]\big] \neq \emptyset$ where $j_l: n_{k_1}[j] = n_{l_{k_1}[j]}[j_l]$. 
\end{theorem}

\begin{proof}
 Constraints \cref{eq:cont_opt_plat_arr_con,eq:cont_opt_plat_trav_con} enforce $t_{k_1}[j] = t_{k_2}[j_l]$ and $t_{k_1}[j+1] = t_{k_2}[j_l+1]$, with $j$ and $j_l$ as in the theorem. On the other hand, we have $t_{k_1}[j] \in \big[\underline{t}_{k_1}[j], \bar{t}_{k_1}[j]\big]$ and $t_{k_2}[j_l] \in \big[\underline{t}_{k_2}[j_l], \bar{t}_{k_2}[j_l]\big]$. So, if $\big[\underline{t}_{k_1}[j], \bar{t}_{k_1}[j]\big] \cap \big[\underline{t}_{k_2}[j_l], \bar{t}_{k_2}[j_l]\big] = \emptyset$, we cannot find a solution to Problem~\ref{prop:convex_cont_problem} with $t_{k_1}[j] = t_{k_2}[j_l]$. This proves necessity. 
 
 Assume $l_k[p] = k$ for $(k,p) \neq (k_1,j),(k_1,j+1)$. If $\big[\underline{t}_{k_1}[j], \bar{t}_{k_1}[j]\big] \cap \big[\underline{t}_{k_2}[j_l], \bar{t}_{k_2}[j_l]\big] \neq \emptyset \wedge \big[\underline{t}_{k_1}[j+1], \bar{t}_{k_1}[j+1]\big] \cap \big[\underline{t}_{k_2}[j_l+1], \bar{t}_{k_2}[j_l+1]\big] \neq \emptyset$, then there exists $T_{k_1}[1], \dots, T_{k_1}[j-1]$ and $T_{k_2}[1], \dots, T_{k_2}[j_l-1]$ such that $\op{max}(\underline{t}_{k_1},\underline{t}_{k_2}) \leq t_{k_1}^S + \sum\limits_{p = 1}^{j-1} T_{k_1}[p] = t_{k_2}^S + \sum\limits_{p = 1}^{j_l-1} T_{k_2}[p] \leq \op{min}(\bar{t}_{k_1},\bar{t}_{k_2})$. Assign now $T_{k}[p] = W(e_{k}[p])/v_{\max}$ for $(k,p) \neq (k_1,1),\dots,(k_1,i-1),(k_2,1),\dots,(k_2,j-1)$. Since the solution for $l_k[i] = k$ exists, these $T_1,\dots,T_K$ fulfill constraints \cref{eq:cont_opt_v_max_con,eq:cont_opt_t_d_con,eq:cont_opt_plat_arr_con,eq:cont_opt_plat_trav_con}. This proves sufficiency. 
\end{proof}
Theorem~\ref{prop:possible_l_per_link} allows us to easily identify possible values for the elements of the sequences $l_k$, $k \in \{1,\dots,K\}$. However, Problem~\ref{prop:convex_cont_problem} does not have a solution for each combination of these values. Therefore, we describe next how to check feasibility for a particular platoon configuration. We do this by iterating over all paths of the trucks and the nodes therein. We calculate the intersection of the feasible arrival time at this node between a truck and its predecessor on the next edge on the truck's path. In each iteration, we set $\big[\underline{t}_{k}[j], \bar{t}_{k}[j]\big]^+ = \big[\underline{t}_{l_k[j]}[j_l], \bar{t}_{l_k[j]}[j_l]\big]^+ = \big[\underline{t}_{k}[j], \bar{t}_{k}[j]\big] \cap \big[\underline{t}_{l_k[i]}[j_l], \bar{t}_{l_k[i]}[j_l]\big]$, where the superscript ``+'' indicates the values after the iteration and $j_l: n_k[j] = n_{l_k[j]}[j_l]$ is the index of node $n_k[j]$ in the path of the predecessor of truck $k$ on edge $e_k[j]$. Then, we propagate constraint \eqref{eq:cont_opt_v_max_con} for the trucks associated with transport assignments $k$ and $l_k[j]$. In this way, we prune the intervals created by $\underline{t}$, $\bar{t}$ until there either is a node where $\big[\underline{t}_{k}[j], \bar{t}_{k}[j]\big] \cap \big[\underline{t}_{l_k[j]}[j_l], \bar{t}_{l_k[j]}[j_l]\big]$ is empty or there is no change after an iteration over all nodes. In the first case, the platoon configuration is not feasible. In the second case, there is a solution to Problem~\ref{prop:convex_cont_problem}. Furthermore, the sequences $t_k, k \in \{1,\dots,K\}$ of the solution lie in the intervals with lower bound $\underline{t}_k$ and upper bound $\bar{t}_k$. 

The above procedure will give us all platoon configurations for which Problem~\ref{prop:convex_cont_problem} has a solution. Recall that we are only interested in platoon configurations where the predecessor of a truck on an edge is the truck in the platoon with next largest index.
However, we want to find the solution to Problem~\ref{prob:shortest_path_problem}. To this end, we can show that in an optimal platoon configuration two trucks will meet and split up at most once.
\begin{theorem}\label{prop:only_one_split}
If the solution to Problem~\ref{prop:convex_cont_problem} with platoon configuration $l_1,\dots,l_k$ is the unique solution to Problem~\ref{prob:shortest_path_problem}, then $l_1,\dots,l_k$ have the property that for $k_l \neq k$ there exists for each $k \in \{1,\dots,K\}$ at most one pair $(j_1,j_2): j_1 \leq j_2$ such that $k_l \neq l_k[j_1-1], k_l = l_k[j_1]$ and $k_l = l_k[j_2], k_l \neq l_k[j_2+1]$.
\end{theorem}

\begin{proof}
 Let $\tilde{n}$ be the path formed by $n_{k_1} \cap_\set{E} n_{k_2}$ and let $\tilde{e}$ be the corresponding sequence of edges. We know from Lemma~\ref{prop:path_intesection} that $\tilde{n}$ will indeed be a path. Assume there is a unique optimal solution with $v_{k_1}$, $v_{k_2}$ to Problem~\ref{prob:shortest_path_problem} where the above statement does not hold, i.e., assume there exists $j_1,j_2: j_2 > j_1$ so that trucks $k_1$, $k_2$ platoon on $(\tilde{n}[j_1-1],\tilde{n}[j_1])$ and on $(\tilde{n}[j_2],\tilde{n}[j_2+1])$ but not on the edges between $\tilde{n}[j_1]$ and $\tilde{n}[j_2]$. Then, $k_1$, $k_2$ could also platoon on the edges between $\tilde{n}[j_1]$ and $\tilde{n}[j_2]$. Denote the speed profile of truck $k_1$, $k_2$ on these edges as $v_{k_1}$, $v_{k_2}$ respectively. 
 Assume without loss of generality $\sum\limits_{m} W(\tilde{e}[j_1-1+m])v_{k_1}[m]^2 \leq \sum\limits_{m} W(\tilde{e}[j_1-1+m]) v_{k_2}[m]^2$. We assign truck $k_2$ on these nodes the speed profile of truck $k_1$ which implies they platoon on these edges. For each edge $m$ there are two cases: $k_2$ was part of another platoon. Then the change in the contribution of this edge to the objective function is $\eta(v_{k_1}[m]^2 - v_{k_2}[m]^2)$. If $k_2$ was not part of a platoon then the change for this edge in the objective function is $\eta v_{k_1}[m]^2 - v_{k_2}[m]^2 \leq \eta(v_{k_1}[m]^2 - v_{k_2}^2)$. So, the change in the objective function from the old to the new speed profile is upper bounded by $\eta \sum\limits_{m} W(\tilde{e}[j_1-1+m])(v_{k_1}[m]^2 - v_{k_2}[m]^2) \leq 0$ which contradicts the assumption that $v_{k_1}$, $v_{k_2}$ are part of a unique optimal solution.
\end{proof}

In the next example, we illustrate how our method is applied to a realistic scenario.
\begin{example}
We have four transport assignments. The road network $\set{G}$ is shown in Figure~\ref{fig:companion_problem_map}. It is a handmade abstraction of an area in the east of Hungary and the west of Romania. The weights are in kilometers. 

We set $v_{\max} = 90$ and $\eta = 0.6$. The start times are set to $t_1^S = 0.2,t_2^S=-0.3,t_3^S=0.75,t_4^s=-0.05$. The arrival deadlines are set to the arrival time when traveling the whole path with a speed of $80$. This creates 340 valid platoon configurations. The result of the optimization is shown in Figure~\ref{fig:companion_problem_result_3d}. The best configuration turns out to be $l_1 = (1, 1, 2, 2, 2, 2, 1)$, $l_2 = (2, 2, 2, 2, 2, 3, 2, 2, 2)$, $l_3 = (3, 3, 4, 4)$, $l_4 = (4, 4, 4)$. It is interesting to note that trucks 1 and 2 (blue and green) do not platoon on the edges $(2,3)$ and $(16,18)$, even though they could according to the constraints. The value of the objective function is 5\,\% lower than the solution where no coordination takes place, i.e., the air drag coefficient is in average reduced by 5\,\%. If more favorable starting times are chosen, the reduction can be much bigger.

\begin{figure}
  \begin{center}
 \includegraphics[width = .98\columnwidth]{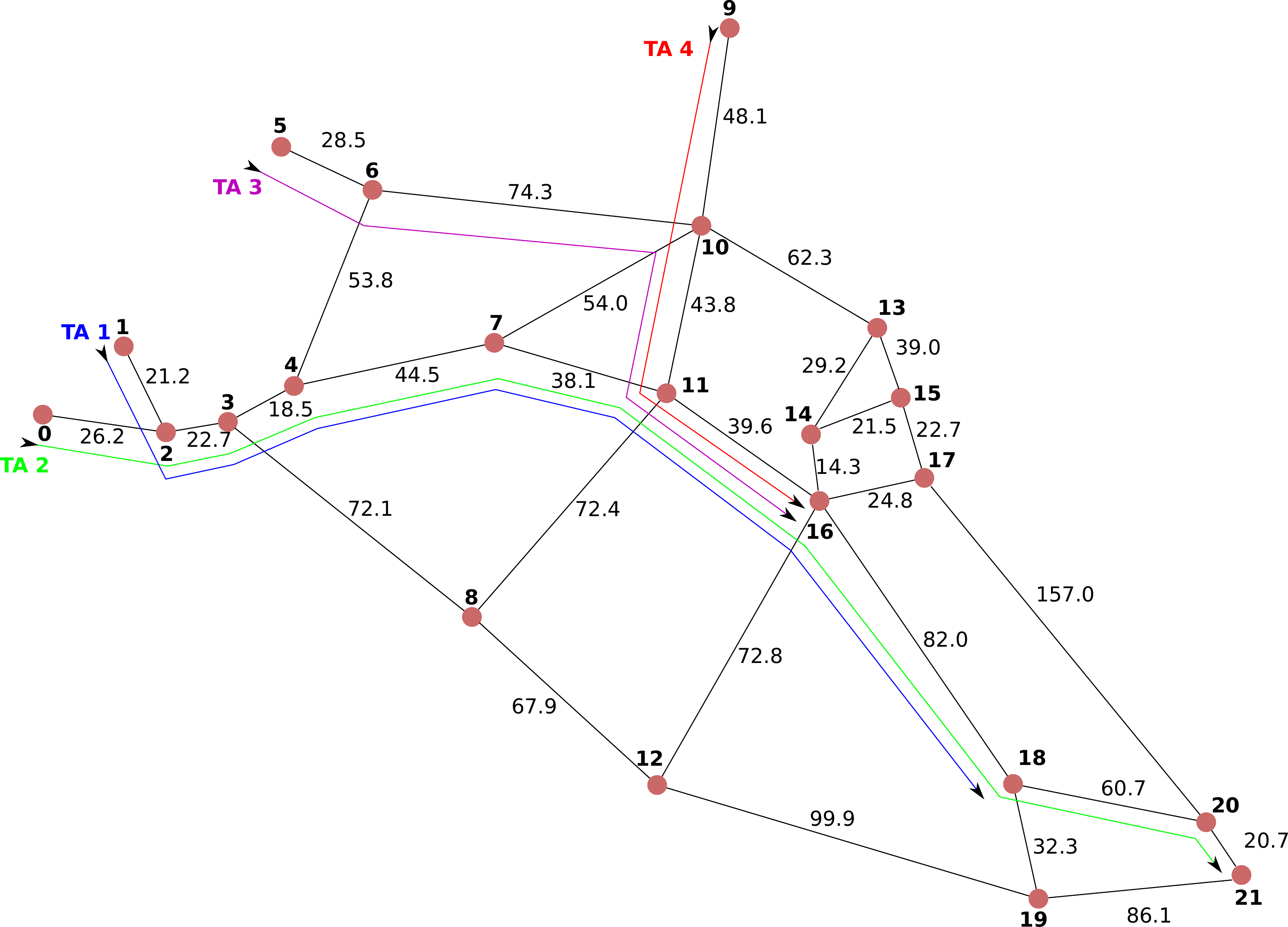}
 \caption{Road-network to illustrate the approximate solution based on shortest paths. The nodes are shown as red dots and the edges as black lines. All edges are bi-directional. The index of a node is placed next to the node and the weight of an edge is placed next to the edge. The shortest paths for the four transport assignments are indicated in color. At the beginning of each path, ``TA'' and the index of the transport assignment are shown.}
\label{fig:companion_problem_map}
\end{center}
\end{figure}

\begin{figure}
  \begin{center}
 \includegraphics[width = .8\columnwidth]{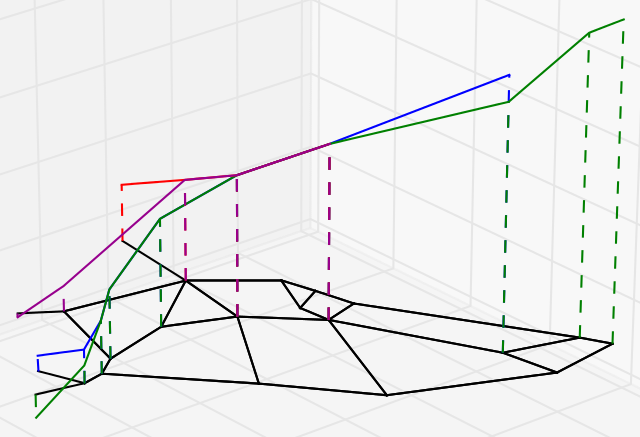}
 \caption{Optimization result of the approximate solution based on shortest paths. In black, the road network of Figure~\ref{fig:companion_problem_map} is shown. The vertical axis indicates the sequences $t_k$ for the four routes at the nodes. The colors are as in Figure~\ref{fig:companion_problem_map}. The dashed lines indicate to which node each data point belongs.}
\label{fig:companion_problem_result_3d}
\end{center}
\end{figure}
 
\end{example}

\section{Conclusion}

In this paper, we have investigated the problem of how to plan the speed of trucks by exploiting the possibility of forming platoons. We ensured that speed limitations and arrival deadlines are not exceeded. We proposed an approximate solution, where each truck takes its shortest path, yielding a number of convex optimization problems. We demonstrated our solution on a realistic scenario.

We made a number of simplifications that would need to be considered in a real world deployment. Some of them can be easily tackled, such as different speed levels per edge and fuel consumption depending on the edge. Other factors, such as unknown traffic conditions and mandatory rest times for the driver, might be more difficult to handle and will be subject to future work. Furthermore, we will investigate how to decrease complexity. The complexity of the algorithm can be problematic in scenarios where many trucks can platoon. 

\vspace{\stretch{1}}

\bibliography{citations}
\bibliographystyle{IEEEtran}

\end{document}